\documentclass[letterpaper, 10 pt, conference]{ieeeconf}
\IEEEoverridecommandlockouts
\let\labelindent\relax
\usepackage{enumitem}
\usepackage{balance}
\usepackage[font={small}]{caption}
\usepackage{subcaption}
\usepackage{array}
\usepackage{textcomp}
\usepackage{mathtools, nccmath}
\usepackage{graphicx}
\usepackage{amsfonts}
\usepackage{amsmath}
\usepackage{amssymb}
\usepackage{algorithm}
\usepackage{algorithmic}
\usepackage{hyperref}
\usepackage{tikz}
\usepackage{arydshln}
\usepackage{multirow}
\usepackage{bm}
\usepackage{epstopdf}
\usepackage{cite}
\usepackage{siunitx}
\usepackage{enumitem}

\usepackage{amsthm}

\newtheorem{theorem}{Theorem}
\newtheorem{lemma}{Lemma}
\newtheorem{corollary}{Corollary}
\newtheorem{remark}{Remark}
\theoremstyle{definition}
\newtheorem{definition}{Definition}
\newtheorem{problem}{Problem}

\title{\LARGE \bf
    Safety-Critical Control with Guaranteed Lipschitz Continuity via Filtered Control Barrier Functions}

\author{Shuo Liu$^{1}$, Wei Xiao$^{2}$ and Calin A. Belta$^{3}$
\thanks{This work was supported in part by the NSF under grant IIS-2024606 at Boston University.}
\thanks{$^{1}$S. Liu is with the department of Mechanical Engineering, Boston
University, Brookline, MA, USA. 
        {\tt\small liushuo@bu.edu}}%
\thanks{$^{2}$W. Xiao is with the Computer Science and Artificial Intelligence Lab, Massachusetts Institute of Technology, Cambridge, MA, USA 
        {\tt\small weixy@mit.edu}}%
\thanks{$^{3}$C. Belta is with the Department of Electrical and Computer Engineering and the Department of Computer Science, University of Maryland, College Park, MD, USA 
        {\tt\small cbelta@umd.edu}}%
}

\begin{document} 
\maketitle

\begin{abstract}
In safety-critical control systems, ensuring both system safety and smooth control input is essential for practical deployment. Existing Control Barrier Function (CBF) frameworks, especially High-Order CBFs (HOCBFs), effectively enforce safety constraints, but also raise concerns about the smoothness of the resulting control inputs. While smoothness typically refers to continuity and differentiability, it does not by itself ensure bounded input variation. In contrast, Lipschitz continuity is a stronger form of continuity that not only is necessary for the theoretical guarantee of safety, but also bounds the rate of variation and eliminates abrupt changes in the control input. Such abrupt changes can degrade system performance or even violate actuator limitations, yet current CBF-based methods do not provide Lipschitz continuity guarantees. This paper introduces Filtered Control Barrier Functions (FCBFs), which extend HOCBFs by incorporating an auxiliary dynamic system—referred to as an input regularization filter—to produce Lipschitz continuous control inputs. The proposed framework ensures safety, control bounds, and Lipschitz continuity of the control inputs simultaneously by integrating FCBFs and HOCBFs within a unified quadratic program (QP). Theoretical guarantees are provided and simulations on a unicycle model demonstrate the effectiveness of the proposed method compared to standard and smoothness-penalized HOCBF approaches.
\end{abstract}

\section{Introduction}
\label{sec:Introduction}

Safety is a fundamental concern in the design and operation of autonomous systems. To address this, a large body of research has incorporated safety constraints into optimal control formulations through the use of Barrier Functions (BFs) and Control Barrier Functions (CBFs). Originally developed in the context of optimization \cite{boyd2004convex}, BFs are Lyapunov-like functions \cite{tee2009barrier} that have been widely used to prove set invariance \cite{aubin2011viability, prajna2007framework} and to design multi-objective or multi-agent control strategies \cite{panagou2013multi, wang2016multi, glotfelter2017nonsmooth}. 

CBFs extend BFs to enforce forward invariance of safe sets for affine control systems. If a CBF satisfies Lyapunov-like conditions, safety is guaranteed \cite{ames2016control}. Combining CBFs with Control Lyapunov Functions (CLFs), the CBF-CLF-QP framework formulates safe, optimal control as a sequence of Quadratic Programs (QPs) \cite{ames2012control, ames2016control}. While initially limited to constraints with relative degree one, extensions like Exponential CBFs \cite{nguyen2016exponential} and High-Order CBFs (HOCBFs) \cite{xiao2021high} now accommodate higher-order constraints. CBF-CLF-QP has been successfully applied in various domains, including rehabilitation \cite{isaly2020zeroing}, adaptive cruise control \cite{liu2023auxiliary,liu2025auxiliary}, humanoid locomotion \cite{khazoom2022humanoid}, and obstacle avoidance involving complex environments \cite{liu2024iterative,liu2024safety}.

Despite these advances, the CBF-CLF-QP framework requires smooth CBFs to compute Lie derivatives, and in practice also relies on smooth control inputs to mitigate chattering phenomena \cite{huber2016implicit, utkin2006chattering, xiao2021high2}. Building on this motivation, several recent works have focused on constructing smooth CBFs as the foundation for safety filters, with the aim of producing smooth control inputs. For instance, in \cite{taylor2022safe}, the authors introduced a barrier backstepping methodology for systems in strict-feedback form, which enables the recursive construction of smooth composite CBFs by differentiating through intermediate virtual controllers, thereby maintaining the continuity and differentiability of the resulting control law. Similarly, \cite{molnar2021model} applied these principles to robotic systems by leveraging reduced-order models, facilitating the design of smooth, real-time implementable CBFs for high-dimensional platforms such as legged robots and aerial vehicles. Additionally, \cite{cohen2023characterizing} provided a rigorous framework for designing smooth safety filters using the implicit function theorem, formally ensuring the existence of differentiable control laws that enforce safety constraints while minimizing deviation from nominal performance. However, it is generally impossible to guarantee smoothness in CBF-based formulations, and moreover a smooth control input does not by itself ensure bounded variation. This can still result in abrupt input changes, which accelerate actuator wear and undermine the validity of system modeling assumptions.


From a theoretical perspective, what is fundamentally required is Lipschitz continuity, which ensures well-posed closed-loop dynamics and underlies safety guarantees. Lipschitz continuity also provides a uniform bound on the rate of input variation, thereby eliminating abrupt changes in the control input. 
This paper introduces a novel class of Filtered Control Barrier Functions (FCBFs), which augment existing HOCBF frameworks with an auxiliary dynamic system acting as an input regularization filter. This auxiliary system takes the original (unfiltered) control input from standard HOCBF formulations and outputs a  Lipschitz continuous filtered control input for the underlying system. The FCBFs are specifically constructed to evolve along the dynamics of this auxiliary system, ensuring that the filtered input respects control bounds while still guaranteeing safety—all without modifying the original safety constraints. The resulting formulation remains a single QP, maintaining computational efficiency and real-time capability. 

The remainder of the article is organized as follows. In
Sec. \ref{sec:Preliminaries}, we provide definitions and preliminaries of HOCBFs and CLFs, and then formulates the problem and outlines the proposed approach in Sec. \ref{sec:Problem Formulation and Approach}. The proposed FCBFs and their theoretical properties are introduced in Sec. \ref{sec:FCBF} followed by
simulations in Sec. \ref{sec:Case Study and Simulations}. We conclude the paper and discuss directions for future work in Sec. \ref{sec:conclusion}.

\section{Definitions and Preliminaries}
\label{sec:Preliminaries}

Consider an affine control system of the form
\begin{equation}
\label{eq:affine-control-system}
\dot{\boldsymbol{x}}=f(\boldsymbol{x})+g(\boldsymbol{x})\boldsymbol{u},
\end{equation}
 where $\boldsymbol{x}\in \mathbb{R}^{n}, f:\mathbb{R}^{n}\to\mathbb{R}^{n}$ and $g:\mathbb{R}^{n}\to\mathbb{R}^{n\times q}$ are locally Lipschitz, and $\boldsymbol{u}\in \mathcal U\subset \mathbb{R}^{q}$, where $\mathcal U$ denotes the control limitation set, which is assumed to be in the form: 
\begin{equation}
\label{eq:control-constraint}
\mathcal U \coloneqq \{\boldsymbol{u}\in \mathbb{R}^{q}:\boldsymbol{u}_{min}\le \boldsymbol{u} \le \boldsymbol{u}_{max} \}, 
\end{equation}
with $\boldsymbol{u}_{min},\boldsymbol{u}_{max}\in \mathbb{R}^{q}$ (vector inequalities are interpreted componentwise). We assume that no component of $\boldsymbol{u}_{min}$ and $\boldsymbol{u}_{max}$ can be infinite. 

\begin{definition}[Lipschitz continuity~\cite{rockafellar1998variational}]
\label{def:Lipschitz continuity}
A function $h: \mathbb{R}^n \to \mathbb{R}^m$ is said to be Lipschitz continuous on a set $\mathcal{D} \subseteq \mathbb{R}^n$ if there exists a constant $L \geq 0$ such that
\begin{equation}
\| h(\mathbf{x}) - h(\mathbf{y}) \| \leq L \| \mathbf{x} - \mathbf{y} \|, \quad \forall \mathbf{x},\mathbf{y} \in \mathcal{D}.
\end{equation}
The smallest such $L$ is called the Lipschitz constant of $h$.
\end{definition}
 Unlike $C^1$ smoothness, which typically refers to continuity and differentiability, Lipschitz continuity does not require derivatives to exist everywhere. 
It only guarantees that the function’s rate of change is bounded, which is weaker than differentiability but stronger than mere continuity.
\begin{definition}[Class $\kappa$ function~\cite{Khalil:1173048}]
\label{def:class-k-f}
A continuous function $\alpha:[0,a)\to[0,+\infty],a>0$ is called a class $\kappa$ function if it is strictly increasing and $\alpha(0)=0.$
\end{definition}

\begin{definition}
\label{def:forward-inv}
A set $\mathcal C\subset \mathbb{R}^{n}$ is forward invariant for system \eqref{eq:affine-control-system} if its solutions for some $\boldsymbol{u} \in \mathcal U$ starting from any $\boldsymbol{x}(0) \in \mathcal C$ satisfy $\boldsymbol{x}(t) \in \mathcal C, \forall t \ge 0.$
\end{definition}

\begin{definition}
\label{def:relative-degree}
The relative degree of a differentiable function $b:\mathbb{R}^{n}\to\mathbb{R}$ is the minimum number of times we need to differentiate it along dynamics \eqref{eq:affine-control-system} until any component of $\boldsymbol{u}$ explicitly shows in the corresponding derivative. 
\end{definition}
\begin{lemma}[Comparison Lemma, e.g., \cite{Khalil:1173048}]
\label{lem:for-invariance}
Let $y:\mathbb{R}^{n} \to \mathbb{R}$ be a continuously differentiable function and let $\boldsymbol{x}(t)$ be a trajectory of system~\eqref{eq:affine-control-system} over $t \in [t_0, t_1]$. 
If the time function $y(\boldsymbol{x}(t))$ satisfies
$\dot{y}(\boldsymbol{x}(t)) \geq -\alpha(y(\boldsymbol{x}(t)))$ for all $t \in [t_0, t_1]$, 
where $\alpha$ is a class $\kappa$ function of its argument, and $y(\boldsymbol{x}(t_0)) \geq 0$, 
then $y(\boldsymbol{x}(t)) \geq 0$ for all $t \in [t_0, t_1]$.
\end{lemma}

In this paper,  a \textbf{safety requirement} is defined as $b(\boldsymbol{x})\ge0$, and \textbf{safety} is the forward invariance of the set $\mathcal C\coloneqq \{\boldsymbol{x}\in\mathbb{R}^{n}:b(\boldsymbol{x})\ge 0\}$. The relative degree of function $b$ is also referred to as the relative degree of safety requirement $b(\boldsymbol{x}) \ge 0$. 

\subsection{High-Order Control Barrier Functions (HOCBFs)}

For a requirement $b(\boldsymbol{x})\ge0$ with relative degree $m$ and $\psi_{0}(\boldsymbol{x})\coloneqq b(\boldsymbol{x}),$ we define a sequence of functions $\psi_{i}:\mathbb{R}^{n}\to\mathbb{R},\ i\in \{1,...,m\}$ as
\begin{equation}
\label{eq:sequence-f1}
\psi_{i}(\boldsymbol{x})\coloneqq\dot{\psi}_{i-1}(\boldsymbol{x})+\alpha_{i}(\psi_{i-1}(\boldsymbol{x})),\ i\in \{1,...,m\}, 
\end{equation}
where $\alpha_{i}(\cdot ),\ i\in \{1,...,m\}$ denotes a $(m-i)^{th}$ order differentiable class $\kappa$ function. We further set up a sequence of sets $\mathcal C_{i}$ based on \eqref{eq:sequence-f1} as
\begin{equation}
\label{eq:sequence-set1}
\mathcal C_{i}\coloneqq \{\boldsymbol{x}\in\mathbb{R}^{n}:\psi_{i}(\boldsymbol{x})\ge 0\}, \ i\in \{0,...,m-1\}. 
\end{equation}
\begin{definition}[HOCBF~\cite{xiao2021high}]
\label{def:HOCBF}
Let $\psi_{i}(\boldsymbol{x}),\ i\in \{1,...,m\}$ be defined by \eqref{eq:sequence-f1} and $\mathcal C_{i},\ i\in \{0,...,m-1\}$ be defined by \eqref{eq:sequence-set1}. A function $b:\mathbb{R}^{n}\to\mathbb{R}$ is a High-Order Control Barrier Function (HOCBF) with relative degree $m$ for system \eqref{eq:affine-control-system} if there exist $(m-i)^{th}$ order differentiable class $\kappa$ functions $\alpha_{i},\ i\in \{1,...,m\}$ such that
\begin{equation}
\label{eq:highest-HOCBF}
\begin{split}
\sup_{\boldsymbol{u}\in \mathcal U}[L_{f}^{m}b(\boldsymbol{x})+L_{g}L_{f}^{m-1}b(\boldsymbol{x})\boldsymbol{u}+O(b(\boldsymbol{x}))
+\\
\alpha_{m}(\psi_{m-1}(\boldsymbol{x}))]\ge 0,
\end{split}
\end{equation}
$\forall \boldsymbol{x}\in \mathcal C_{0}\cap...\cap \mathcal C_{m-1},$ where $L_{f}^{m}$ denotes the $m^{th}$ Lie derivative along $f$ and $L_{g}$ denotes the matrix of Lie derivatives along the columns of $g$; 
$O(\cdot)=\sum_{i=1}^{m-1}L_{f}^{i}(\alpha_{m-1}\circ\psi_{m-i-1})(\boldsymbol{x})$ contains the remaining Lie derivatives along $f$ with degree less than or equal to $m-1$. $\psi_{i}(\boldsymbol{x})\ge0$ is referred to as the $i^{\text{th}}$-order HOCBF inequality (constraint in optimization). We assume that $L_{g}L_{f}^{m-1}b(\boldsymbol{x})\boldsymbol{u}\ne0$ on the boundary of set $\mathcal C_{0}\cap\ldots\cap \mathcal C_{m-1}.$ 
\end{definition}

\begin{theorem}[Safety Guarantee~\cite{xiao2021high}]
\label{thm:safety-guarantee}
Given a HOCBF $b(\boldsymbol{x})$ from Def. \ref{def:HOCBF} with corresponding sets $\mathcal{C}_{0}, \dots,\mathcal {C}_{m-1}$ defined by \eqref{eq:sequence-set1}, if $\boldsymbol{x}(0) \in \mathcal {C}_{0}\cap \dots \cap \mathcal {C}_{m-1},$ then any Lipschitz controller $\boldsymbol{u}$ that satisfies the inequality in \eqref{eq:highest-HOCBF}, $\forall t\ge 0$ renders $\mathcal {C}_{0}\cap \dots \cap \mathcal {C}_{m-1}$ forward invariant for system \eqref{eq:affine-control-system}, $i.e., \boldsymbol{x} \in \mathcal {C}_{0}\cap \dots \cap \mathcal {C}_{m-1}, \forall t\ge 0.$
\end{theorem}
The Lipschitz controller requirement in Thm.~\ref{thm:safety-guarantee} ensures that the closed-loop vector field is well-posed, guaranteeing the existence and uniqueness of trajectories for system~\eqref{eq:affine-control-system}.
\begin{definition}[CLF~\cite{ames2012control}]
\label{def:control-l-f}
A continuously differentiable function $V:\mathbb{R}^{n}\to\mathbb{R}$ is an exponentially stabilizing Control Lyapunov Function (CLF) for system \eqref{eq:affine-control-system} if there exist constants $c_{1}>0, c_{2}>0,c_{3}>0$ such that for $\forall \boldsymbol{x} \in \mathbb{R}^{n}, c_{1}\left \|  \boldsymbol{x} \right \| ^{2} \le V(\boldsymbol{x}) \le c_{2}\left \|  \boldsymbol{x} \right \| ^{2}$ and
\begin{equation}
\label{eq:clf}
\inf_{\boldsymbol{u}\in \mathcal U}[L_{f}V(\boldsymbol{x})+L_{g}V(\boldsymbol{x})\boldsymbol{u}+c_{3}V(\boldsymbol{x})]\le 0.
\end{equation}
\end{definition}

Several works (e.g., \cite{nguyen2016exponential}, \cite{xiao2021high}) formulate safety-critical optimization problems by combining HOCBFs \eqref{eq:highest-HOCBF} with quadratic cost objectives to handle systems with high relative degree. HOCBFs ensure forward invariance of safety-related sets, thereby guaranteeing safety, while CLFs \eqref{eq:clf} can also be included as soft constraints to enforce exponential convergence of desired states \cite{xiao2021high, liu2025auxiliary}. In these formulations, control inputs are treated as decision variables in a QP, solved at each discrete time step with the current state held fixed. The resulting optimal control is applied at the beginning of each interval and kept constant, while system dynamics \eqref{eq:affine-control-system} are used to update the state. Note that in Thm. \ref{thm:safety-guarantee}, we assume the input is a Lipschitz continuous controller. However, since $\boldsymbol{u}$ is a decision variable rather than a closed-form expression, there is no analytical guarantee that it is Lipschitz continuous. As a result, the control input may vary abruptly over time and can be difficult to regulate.  In this paper, we provide a theoretical guarantee of Lipschitz continuity for the control input, ensuring bounded variation even under piecewise QP-based implementations. Nevertheless, due to the inherent discontinuities introduced by the piecewise update scheme, the proposed Lipschitz controller is not fully smooth. Notably, the proposed method can be applied to any system where conventional CBFs are valid, making it widely applicable, easy to integrate into existing control frameworks \cite{taylor2022safe, molnar2021model,cohen2023characterizing} that focus on improving smoothness.

\section{Problem Formulation and Approach}
\label{sec:Problem Formulation and Approach}

Our goal is to generate a control strategy for system \eqref{eq:affine-control-system} that ensures convergence, minimizes energy, satisfies safety, and respects input constraints. 

\textbf{Objective:} We consider the cost  
\begin{equation}
\label{eq:cost-function-1}
\begin{split}
 J(\boldsymbol{u}(t))=\int_{0}^{T} 
 \| \boldsymbol{u}(t) \| ^{2}dt+p\left \| \boldsymbol{x}(T)-\boldsymbol{x}_{e} \right \| ^{2},
\end{split}
\end{equation}
where $\left \| \cdot \right \|$ denotes the 2-norm of a vector, and $T>0$ is the ending time; $p>0$ denotes a weight factor and $\boldsymbol{x}_{e} \in \mathbb{R}^{n}$ is a desired state, which is assumed to be an equilibrium for system \eqref{eq:affine-control-system}. $p\left \| \boldsymbol{x}(T)-\boldsymbol{x}_{e} \right \| ^{2}$ enforces state convergence.

\textbf{Safety Requirement:} System \eqref{eq:affine-control-system} should always satisfy one or more safety requirements of the form: 
\begin{equation}
\label{eq:Safety constraint}
b(\boldsymbol{x})\ge 0, \boldsymbol{x} \in \mathbb{R}^{n}, \forall t \in [0, T],
\end{equation}
where $b:\mathbb{R}^{n}\to\mathbb{R}$ is assumed to be a continuously differentiable equation. 

\textbf{Control Limitations:} The controller $\boldsymbol{u}$ should always satisfy \eqref{eq:control-constraint} for all $t \in [0, T].$

A control policy is \textbf{feasible} if \eqref{eq:Safety constraint} and \eqref{eq:control-constraint} are satisfied $\forall t \in [0, T].$ In this paper, we consider the following problem:

\begin{problem}
\label{prob:smooth-prob}
Find a feasible Lipschitz continuous control policy for system \eqref{eq:affine-control-system} such that cost \eqref{eq:cost-function-1} is minimized.
\end{problem}

\textbf{Approach:} 
To solve Problem 1, we build on the CBF-based QP framework introduced at the end of Sec. \ref{sec:Preliminaries}. A HOCBF $b_{1}(\boldsymbol{x})$ is used to enforce the safety constraint in \eqref{eq:Safety constraint}, while a CLF handles the terminal state constraint in the cost function \eqref{eq:cost-function-1}. To ensure that the control input involved in the $m^{\text{th}}$-order constraint $\psi_{m}(\boldsymbol{x},\boldsymbol{u})\ge 0$ is Lipschitz continuous as required by Thm. \ref{thm:safety-guarantee}, we introduce a filtered input $\boldsymbol{u}_{f}$, generated by an auxiliary dynamic system that is an input regularization filter by itself. The filtered input will replace original input $\boldsymbol{u}$ in $\psi_{m}$. We then define a new HOCBF, $b_{2}(\boldsymbol{x},\boldsymbol{u}_{f})=\psi_{m}(\boldsymbol{x},\boldsymbol{u}_{f})$, where $\boldsymbol{u}_{f}$ is treated as part of the state of the auxiliary system. This formulation leads to a set of high-order constraints that jointly guarantee safety, input bounds \eqref{eq:control-constraint}, and Lipschitz continuity of $\boldsymbol{u}_{f}$. We refer to $b_{2}$ as a Filtered CBF (FCBF). Since its structure is analogous to standard HOCBFs, the overall formulation remains a QP, preserving both structural simplicity and computational efficiency. While the resulting control input may exhibit small oscillations, it remains bounded and Lipschitz continuous. The Lipschitz constant defined in Def. \ref{def:Lipschitz continuity} can be tuned via design hyperparameters to further reduce these oscillations for practical implementation and theoretical guarantees.

\section{Filtered Control Barrier Functions}
\label{sec:FCBF}
\subsection{Motivating Example}
\label{subsec:Motivation Example}

Consider a simplified unicycle model defined in the form:
\begin{equation}
\label{eq:simplified unicycle model}
\begin{split}
 \dot{x}=v\cos (\theta), \dot{y}=v\sin (\theta),
 \dot{\theta}=u,
\end{split}
\end{equation}
where $v$ is a constant linear speed, $\boldsymbol{x}=[x,y,\theta]^{\top}$, $[x,y]^{\top}\in \mathbb{R}^{2}$ denotes the 2-D location of the vehicle, $\theta \in \mathbb{R}$ denotes the heading, and $u$ is the control input corresponding to steering wheel angle. Suppose system \eqref{eq:simplified unicycle model} has to satisfy a safety constraint:
\begin{equation}
\label{eq:safety constraint 1}
\begin{split}
(x-x_{o})^{2}+(y-y_{o})^{2}-r_{o}^{2}\ge 0,
\end{split}
\end{equation}
where $[x_{o}, y_{o}]^{\top}\in \mathbb{R}^{2}$ denotes the 2-D location of a circular obstacle, and $r_{o}>0$ denotes its radius. The safety constraint in \eqref{eq:safety constraint 1} has relative degree two with respect to the system dynamics in \eqref{eq:simplified unicycle model}. Therefore, a HOCBF with $m=2$, as defined in Def. \ref{def:HOCBF}, can be used to enforce this constraint. We choose the class $\kappa$ functions $\alpha_{1}, \alpha_{2}$ as linear functions and the corresponding HOCBF constraint \eqref{eq:highest-HOCBF} becomes:
\begin{equation}
\label{eq:h-o safety constraint 1}
\begin{split}
(2(y-y_{o})v\cos{\theta}-2(x-x_{o})v\sin{\theta})u+2v^{2}+\\(k_{1}+k_{2})\dot{b}(\boldsymbol{x})+k_{1}k_{2}b(\boldsymbol{x})\ge0,
\end{split}
\end{equation}
where $b(\boldsymbol{x})=(x-x_{o})^{2}+(y-y_{o})^{2}-r_{o}^{2}$ and $\dot{b}(\boldsymbol{x})=2(x-x_{o})v\cos(\theta)+2(y-y_{o})v\sin(\theta)$. Note that the expression $\frac{2v^{2}+\\(k_{1}+k_{2})\dot{b}(\boldsymbol{x})+k_{1}k_{2}b(\boldsymbol{x})}{2(y-y_{o})v\cos{\theta}-2(x-x_{o})v\sin{\theta}}$ can be differentiable, yet its derivative can be unbounded. This can cause the control input $u$ to vary rapidly over time in order to satisfy the high-order safety constraint \eqref{eq:h-o safety constraint 1}, potentially leading to abrupt or nonsmooth behavior. In such cases, the assumption in Thm. \ref{thm:safety-guarantee} that $u$ is Lipschitz continuous may no longer hold. This problem is more serious for systems with noisy dynamics. In this work, we address this issue by proposing a method that ensures the Lipschitz continuity of the control input $u$.

\subsection{Filtered Control Barrier Functions (FCBF)}

Motivated by the simplified unicycle example in Sec. \ref{subsec:Motivation Example}, 
Given a HOCBF $b:\mathbb{R}^{n}\to\mathbb{R}$ with relative degree $m$ for system \eqref{eq:affine-control-system} and control limitations \eqref{eq:control-constraint}, we aim to develop a method that ensures the control input is Lipschitz continuous, allowing it to change without abrupt variations over time while also satisfying control limitations. In order to achieve this, we need to define the filtered control input $\boldsymbol{u}_f$, which will be used to replace the original input in the dynamic system \eqref{eq:affine-control-system}, i.e., $\bm u = \bm u_f$. We also need to impose constraints on the time derivative of the filtered input $\boldsymbol{u}_f$. Note that the authors of \cite{dacs2024rollover} impose constraints on the input by introducing a feedback law and a differentiator, while the authors of \cite{ames2020integral} impose constraints on the first-order time derivative of the input via a feedback law and an auxiliary input. However, such methods cannot control higher-order time derivatives of the input, nor can they guarantee its Lipschitz continuity. To overcome these shortcomings, we first construct an auxiliary dynamical system that introduces higher-order time derivatives of $\boldsymbol{u}_f$. Since this auxiliary system is designed specifically to apply bounded-rate regularization on $\boldsymbol{u}_f$ over time, we refer to it as an input regularization filter:
\begin{equation}
\label{eq:smoothness filter}
\dot{\boldsymbol{\pi}}=F(\boldsymbol{\pi})+G(\boldsymbol{\pi})\boldsymbol{\nu},
\end{equation}
where $\boldsymbol{\pi}(t)\coloneqq [\pi_{1}(t),\dots,\pi_{m_{a}q}(t)]^{^{\top}}\in \mathbb{R}^{m_{a}q}$  denotes an auxiliary state with $\pi_{j}(t)\in \mathbb{R}, j \in \{1,...,m_{a}q\}$. $\nu \in \mathbb{R}^{q}$ is an auxiliary input for \eqref{eq:smoothness filter}, $F:\mathbb{R}^{m_{a}q}\to\mathbb{R}^{m_{a}q}$ and $G:\mathbb{R}^{m_{a}q}\to\mathbb{R}^{m_{a}q\times q}$ are locally Lipschitz. For simplicity, we just build up the connection between auxiliary variables and the system as $\boldsymbol{\pi}_{1:q}=\boldsymbol{u}_f,\boldsymbol{\pi}_{(q+1):2q}=\boldsymbol{u}^{(1)}_f,\dots, \boldsymbol{\pi}_{((m_{a}-1)q+1):m_{a}q}=\boldsymbol{u}^{(m_{a}-1)}_f,\boldsymbol{u}^{(m_{a})}_f=\boldsymbol{\nu}$. $\boldsymbol{\pi}_{i:j}$ refers to the entries of vector $\boldsymbol{\pi}$ from index $i$ to $j$ and $\boldsymbol{u}_f^{(i)}$ represents the $i^{\text{th}}$ derivative of $\boldsymbol{u}_f$ with respect to time $t$. A more general construction of system \eqref{eq:smoothness filter} requires that the auxiliary state $\boldsymbol{u}_{f}$ has relative degree $m_{a}$ with respect to system \eqref{eq:smoothness filter}. In order to guarantee safety, we first define $\psi_{0,f}(\boldsymbol{x},\boldsymbol{\pi})\coloneqq \psi_{m}(\boldsymbol{x},\boldsymbol{u}_f)$ based on \eqref{eq:highest-HOCBF}, where the relative degree of $\psi_{0,f}(\boldsymbol{x},\boldsymbol{\pi})$ with respect to system \eqref{eq:smoothness filter} is $m_{a}$. We then define a sequence of functions $\psi_{i,f}:\mathbb{R}^{n}\to\mathbb{R},\ i\in \{1,...,m_{a}\}$ as
\begin{equation}
\label{eq:sequence-f2}
\begin{split}
\psi_{i,f}(\boldsymbol{x},\boldsymbol{\pi})\coloneqq\dot{\psi}_{i-1,f}(\boldsymbol{x},\boldsymbol{\pi})+\alpha_{i}(\psi_{i-1,f}(\boldsymbol{x},\boldsymbol{\pi})),
\end{split}
\end{equation}
where $\alpha_{i}(\cdot ),\ i\in \{1,...,m_{a}\}$ denotes a $(m_{a}-i)^{th}$ order differentiable class $\kappa$ function. We further set up a sequence of sets $\mathcal C_{i,f}, i\in \{0,...,m_{a}-1\}$ based on \eqref{eq:sequence-f2} as
\begin{equation}
\label{eq:sequence-set2}
\begin{split}
\mathcal C_{i,f}\coloneqq \{(\boldsymbol{x},\boldsymbol{\pi})\in\mathbb{R}^{n+m_{a}q}:\psi_{i,f}(\boldsymbol{x},\boldsymbol{\pi})\ge 0\} . 
\end{split}
\end{equation}

\begin{definition}[FCBF]
\label{def:FCBF}
Let $\psi_{i,f}(\boldsymbol{x}),\ i\in \{1,...,m_{a}\}$ be defined by \eqref{eq:sequence-f2}, $\mathcal C_{i,f},\ i\in \{0,...,m_{a}-1\}$ be defined by \eqref{eq:sequence-set2}, and a function $b(\boldsymbol{x}):\mathbb{R}^{n}\to\mathbb{R}$ is a HOCBF with relative degree $m$ for system \eqref{eq:affine-control-system} defined by Def. \ref{def:HOCBF}. Function $b_{f}(\boldsymbol{x},\boldsymbol{\pi})=\psi_{0,f}(\boldsymbol{x},\boldsymbol{\pi})=\psi_m(\bm x, \bm u_f)$, where $\bm u_f = \bm u$, is a Filtered Control Barrier Function (FCBF) with relative degree $m_{a}$ for system \eqref{eq:smoothness filter} if there exist $(m_{a}-i)^{th}$ order differentiable class $\kappa$ functions $\alpha_{i},\ i\in \{1,...,m_{a}\}$ such that
\begin{equation}
\label{eq:highest-FCBF}
\begin{split}
\sup_{\boldsymbol{\nu}\in \mathbb{R}^{q}}[L_{F}^{m_{a}}b_{f}(\boldsymbol{x},\boldsymbol{\pi})+L_{G}L_{F}^{m_{a}-1}b_{f}(\boldsymbol{x},\boldsymbol{\pi})\boldsymbol{\nu}+O(b_{f}(\boldsymbol{x},\boldsymbol{\pi}))
+\\
\alpha_{m_{a}}(\psi_{m_{a}-1,f}(\boldsymbol{x},\boldsymbol{\pi}))]\ge 0,
\end{split}
\end{equation}
$\forall (\boldsymbol{x},\boldsymbol{\pi})\in \mathcal C_{0,f}\cap...\cap \mathcal C_{m_{a}-1,f},$ where $L_{F}^{m_{a}}$ denotes the $m_{a}^{th}$ Lie derivative along $F$ and $L_{G}$ denotes the matrix of Lie derivatives along the columns of $G$; 
$O(\cdot)=\sum_{i=1}^{m_{a}-1}L_{F}^{i}(\alpha_{m_{a}-1}\circ\psi_{m_{a}-i-1,f})(\boldsymbol{x},\boldsymbol{\pi})$ contains the remaining Lie derivatives along $F$ with degree less than or equal to $m_{a}-1$. $\psi_{i,f}(\boldsymbol{x},\boldsymbol{\pi})\ge0$ is referred to as the $i^{\text{th}}$-order FCBF inequality (constraint in optimization). We assume that $L_{G}L_{F}^{m_{a}-1}b_{f}(\boldsymbol{x},\boldsymbol{\pi})\boldsymbol{\nu}\ne0$ on the boundary of set $\mathcal C_{0,f}\cap\ldots\cap \mathcal C_{m_{a}-1,f}$. 
\end{definition}
\begin{theorem}[Safety Guarantee]
\label{thm:safety-guarantee-2}
Given a FCBF $b_{f}(\boldsymbol{x},\boldsymbol{\pi})$ from Def. \ref{def:FCBF} with corresponding sets $\mathcal{C}_{0,f}, \dots,\mathcal {C}_{m_{a}-1,f}$ defined by \eqref{eq:sequence-set2}, if $(\boldsymbol{x}(0),\boldsymbol{\pi}(0)) \in \mathcal {C}_{0}\cap \dots \cap  \mathcal {C}_{m-1}\cap  \mathcal {C}_{0,f}\cap \dots \cap \mathcal {C}_{m_{a}-1,f},$ then any Lipschitz controller $\boldsymbol{\nu}$ that satisfies the inequality in \eqref{eq:highest-FCBF}, $\forall t\ge 0$ renders $\mathcal {C}_{0}\cap \dots \cap  \mathcal {C}_{m-1}\cap \mathcal {C}_{0,f}\cap \dots \cap \mathcal {C}_{m_{a}-1,f}$ forward invariant for systems \eqref{eq:affine-control-system} and \eqref{eq:smoothness filter}, $i.e., (\boldsymbol{x},\boldsymbol{\pi}) \in \mathcal {C}_{0}\cap \dots \cap  \mathcal {C}_{m-1}\cap \mathcal {C}_{0,f}\cap \dots \cap \mathcal {C}_{m_{a}-1,f}, \forall t\ge 0$.
\end{theorem}
\begin{proof}
Since $\boldsymbol{\nu}$ is Lipschitz continuous and appears only in the last equation of \eqref{eq:sequence-f2} when taking Lie derivatives of \eqref{eq:sequence-f2}, it follows that 
$\psi_{m_{a},f}(\boldsymbol{x},\boldsymbol{\pi})$ is also Lipschitz continuous. Moreover, all system states in \eqref{eq:affine-control-system} and \eqref{eq:smoothness filter} are continuously differentiable, so $\psi_{1}, \psi_{2},\dots,\psi_{m-1}, \psi_{0,f},\psi_{1,f},\dots, \psi_{m_{a}-1,f}$ are also continuously differentiable. Then, $\psi_{m_{a},f}(\boldsymbol{x},\boldsymbol{\pi})\ge0$ for $\forall t \ge 0$, i.e., $\dot{\psi}_{m_{a}-1,f}(\boldsymbol{x},\boldsymbol{\pi})+\alpha_{m_{a}}(\psi_{m_{a}-1,f}(\boldsymbol{x},\boldsymbol{\pi}))\ge 0$. By Lemma \ref{lem:for-invariance}, since $(\boldsymbol{x}(0),\boldsymbol{\pi}(0)) \in \mathcal {C}_{m_{a}-1,f}$, i.e., $\psi_{m_{a}-1,f}(\boldsymbol{x}(0),\boldsymbol{\pi}(0))\ge0$, we have $\psi_{m_{a}-1,f}(\boldsymbol{x},\boldsymbol{\pi})\ge0,  \forall t\ge 0$. By Lemma \ref{lem:for-invariance}, since $(\boldsymbol{x}(0),\boldsymbol{\pi}(0)) \in \mathcal {C}_{m_{a}-2,f}$,  we have $\psi_{m_{a}-2,f}(\boldsymbol{x},\boldsymbol{\pi})\ge0,  \forall t\ge 0$. Repeatedly, we have $\psi_{0,f}(\boldsymbol{x},\boldsymbol{\pi})\ge0,  \forall t\ge 0$. This is equivalent to have $\psi_{m}(\boldsymbol{x},\boldsymbol{u})=\psi_{m}(\boldsymbol{x},\boldsymbol{u}_{f})\ge0,  \forall t\ge 0$ since we replace $\boldsymbol{u}$ with $\boldsymbol{u}_{f}$. By Lemma \ref{lem:for-invariance}, since $\boldsymbol{x}(0) \in \mathcal {C}_{m-1}$, i.e., $\psi_{m-1}(\boldsymbol{x}(0))\ge0$, we have $\psi_{m-1}(\boldsymbol{x})\ge0,  \forall t\ge 0$. Repeatedly, we have $\psi_{0}(\boldsymbol{x})=b(\boldsymbol{x})\ge0,  \forall t\ge 0$. Therefore, the intersection of sets $\mathcal {C}_{0}, \dots, \mathcal {C}_{m-1}, \mathcal {C}_{0,f}, \dots ,\mathcal {C}_{m_{a}-1,f}$ are forward invariant.
\end{proof}
Some issues still need to be addressed, such as how to ensure the Lipschitz continuity of $\boldsymbol{u}_{f}$ after replacing $\boldsymbol{u}$ with $\boldsymbol{u}_{f}$, and how to guarantee that $\boldsymbol{u}_{f}$ satisfies the control limitations in \eqref{eq:control-constraint}, i.e., $\boldsymbol{u}_{\min}\le \boldsymbol{u}_{f} \le \boldsymbol{u}_{\max}$. To address these issues, we propose the next theorem. 
\begin{theorem}
\label{thm:lipschitz-guarantee}
Consider the auxiliary system defined in \eqref{eq:smoothness filter} with auxiliary state vector \( \boldsymbol{u}_f \), and the control constraint given in \eqref{eq:control-constraint}. For each \( \boldsymbol{u}_{k,f} \), where \( k \in \{1, \dots, q\} \) denotes the \( k^\text{th} \) entry of the vector \( \boldsymbol{u}_f \), define the corresponding control bounds \( \boldsymbol{u}_{k,\min} \) and \( \boldsymbol{u}_{k,\max} \) as the \( k^\text{th} \) entries of \( \boldsymbol{u}_{\min} \) and \( \boldsymbol{u}_{\max} \), respectively. Suppose the functions
\begin{equation}
b_k^{\min}(\boldsymbol{\pi}) = \boldsymbol{u}_{k,f} - \boldsymbol{u}_{k,\min},  
 b_k^{\max}(\boldsymbol{\pi}) = \boldsymbol{u}_{k,\max} - \boldsymbol{u}_{k,f},  
\end{equation}
are HOCBFs of relative degree \( m_a \) for system \eqref{eq:smoothness filter}, and that the sequences of functions \( \psi_{k,i}^{\min}(\boldsymbol{\pi}) \) and \( \psi_{k,i}^{\max}(\boldsymbol{\pi}) \) for \( i \in \{1, \dots, m_a\} \) are defined as in \eqref{eq:sequence-f1}, with corresponding sets \( \mathcal{C}_i^{\min} \), \( \mathcal{C}_i^{\max} \) for \( i \in \{0, \dots, m_a - 1\} \) defined by \eqref{eq:sequence-set1}. If the initial condition satisfies $\boldsymbol{\pi}(0) \in \bigcap_{i=0}^{m_a - 1} \mathcal{C}_i^{\min} \cap \mathcal{C}_i^{\max}$, and there always exists Lipschitz controller $\boldsymbol{\nu}(t)$ to satisfy both \( \psi_{k,m_{a}}^{\min}(\boldsymbol{\pi}) \ge 0\) and \( \psi_{k,m_{a}}^{\max}(\boldsymbol{\pi})\ge 0 \),
then the filtered input \( \boldsymbol{u}_f \) is Lipschitz continuous and satisfies the control input bounds $\boldsymbol{u}_{\min} \le \boldsymbol{u}_f \le \boldsymbol{u}_{\max}$ componentwise.
\end{theorem}
\begin{proof}
Since $b_k^{\text{min}}(\boldsymbol{\pi}),b_k^{\text{max}}(\boldsymbol{\pi})$ are HOCBFs, the initial condition satisfies $\boldsymbol{\pi}(0) \in \bigcap_{i=0}^{m_a - 1} \mathcal{C}_i^{\text{min}} \cap \mathcal{C}_i^{\text{max}}$, and there always exists $\boldsymbol{\pi}(t)$ to satisfy both \( \psi_{k,m_{a}}^{\text{min}}(\boldsymbol{\pi}) \ge 0\) and \( \psi_{k,m_{a}}^{\text{max}}(\boldsymbol{\pi})\ge 0 \), based on the Thm. \ref{thm:safety-guarantee}, the forward invariance of $\bigcap_{i=0}^{m_a - 1} \mathcal{C}_i^{\text{min}} \cap \mathcal{C}_i^{\text{max}}$ for system \eqref{eq:smoothness filter} is guaranteed. This means $\boldsymbol{u}_{k,\min} \le \boldsymbol{u}_{k,f} \le \boldsymbol{u}_{k,\max}$ is guaranteed. Based on \eqref{eq:sequence-f1}, we also have 
\begin{equation}
\label{eq:Lipschitz-proof-1}
\begin{split}
  \psi_{k,1}^{\text{min}}(\boldsymbol{\pi}) = \dot{\boldsymbol{u}}_{k,f} + \alpha_{k,1}^{\text{min}}(\boldsymbol{u}_{k,f}-\boldsymbol{u}_{k,\min})\ge 0,  \\
  \psi_{k,1}^{\text{max}}(\boldsymbol{\pi}) = -\dot{\boldsymbol{u}}_{k,f} + \alpha_{k,1}^{\text{max}}(\boldsymbol{u}_{k,\max}-\boldsymbol{u}_{k,f})\ge 0, 
\end{split}
\end{equation}
for $t\ge 0$. Rewrite \eqref{eq:Lipschitz-proof-1}, we have 
\begin{equation}
\label{eq:Lipschitz-proof-2}
\begin{split}
-\alpha_{k,1}^{\text{min}}(\boldsymbol{u}_{k,f}-\boldsymbol{u}_{k,\min}) \le \dot{\boldsymbol{u}}_{k,f}\le \alpha_{k,1}^{\text{max}}(\boldsymbol{u}_{k,\max}-\boldsymbol{u}_{k,f}).
\end{split}
\end{equation}
Rewrite \eqref{eq:Lipschitz-proof-2}, we have 
\begin{equation}
\label{eq:Lipschitz-proof-3}
\begin{split}
\left | \dot{\boldsymbol{u}}_{k,f} \right | \le \max (\left | \alpha_{k,1}^{\text{min}}(\boldsymbol{u}_{k,f}-\boldsymbol{u}_{k,\min}) \right | ,\\
\left | \alpha_{k,1}^{\text{max}}(\boldsymbol{u}_{k,\max}-\boldsymbol{u}_{k,f}) \right |)\le L_{k},
\end{split}
\end{equation}    
where $L_{k} \in \mathbb{R}$ is the uniform-in-time Lipschitz constant, as defined in Def.~\ref{def:Lipschitz continuity}, on the forward-invariant compact set $\boldsymbol{u}_{k, \min} \le \boldsymbol{u}_{k, f}(t) \le \boldsymbol{u}_{k, \max}$. We know that equation \eqref{eq:Lipschitz-proof-3} always holds because class $\kappa$ functions are continuous, and the arguments involved lie within a compact interval as $\boldsymbol{u}_{k,\min} \le \boldsymbol{u}_{k,f} \le \boldsymbol{u}_{k,\max}$. Therefore, $\boldsymbol{u}_{k,f}$ is Lipschitz continuous with respect to time. Since each component $\boldsymbol{u}_{k,f}$ of the vector $\boldsymbol{u}_{f}$ is Lipschitz continuous, it follows that $\boldsymbol{u}_{f}$ is also Lipschitz continuous as a vector-valued function. Moreover, the control input bounds $\boldsymbol{u}_{\min} \le \boldsymbol{u}_f \le \boldsymbol{u}_{\max}$ are satisfied componentwise.
\end{proof}
\begin{corollary}(Consequence of Thms. \ref{thm:safety-guarantee-2} and \ref{thm:lipschitz-guarantee})
\label{cor: FCBF-HOCBF}
Consider system \eqref{eq:affine-control-system} with control input $\boldsymbol{u}$ and control limits given by \eqref{eq:control-constraint}. Let $b(\boldsymbol{x})$ be a safety constraint with relative degree $m$, and $b_f(\boldsymbol{x},\boldsymbol{\pi})$ be a FCBF with relative degree $m_a$ for the auxiliary system \eqref{eq:smoothness filter}. 
Suppose:
\begin{enumerate}[label=\arabic*)]
    \item A controller $\boldsymbol{\nu}$ exists that satisfies the FCBF conditions in Thm. \ref{thm:safety-guarantee-2}, rendering the sets $\mathcal{C}_{0,f}, \dots, \mathcal{C}_{m_a-1,f}$ forward invariant.\label{item:fcbf}
    \item For each entry $\boldsymbol{u}_{k,f}$ of the filtered input $\boldsymbol{u}_f$, the functions $b_k^{\min}(\boldsymbol{\pi}) = \boldsymbol{u}_{k,f} - \boldsymbol{u}_{k,\min}$ and $b_k^{\max}(\boldsymbol{\pi}) = \boldsymbol{u}_{k,\max} - \boldsymbol{u}_{k,f}$ are HOCBFs for the auxiliary system with relative degree $m_a$, satisfying the conditions in Thm. \ref{thm:lipschitz-guarantee}.\label{item:hocbf}
    \item The control input $\boldsymbol{\nu}$ satisfies all constraints in Items \ref{item:fcbf} and \ref{item:hocbf} simultaneously, i.e., the FCBF and HOCBF conditions are jointly feasible and admit a common solution $\boldsymbol{\nu}$ at each time $t$.
\end{enumerate}
Then, the filtered input $\boldsymbol{u}_f$ is Lipschitz continuous and satisfies the safety constraint $b(x) \geq 0$ as well as the input constraints $\boldsymbol{u}_{\min} \leq \boldsymbol{u}_f \leq \boldsymbol{u}_{\max}$ for all $t \ge 0$.
\end{corollary}
\begin{proof}
The result follows directly from Thms. \ref{thm:safety-guarantee-2} and \ref{thm:lipschitz-guarantee}.
\end{proof}
\begin{corollary}
\label{cor:first-order-low pass filter}
Based on Cor. \ref{cor: FCBF-HOCBF}, consider the auxiliary system defined by the first-order low-pass filter:
\begin{equation}
\label{eq:low pass filter}
\dot{\boldsymbol{u}}_f = \frac{1}{\tau} (\boldsymbol{\nu} - \boldsymbol{u}_f),
\end{equation}
with filter hyperparameter \( \tau > 0 \), where \( \boldsymbol{u}_f \in \mathbb{R}^q \) is the filtered control input. Suppose all conditions in Cor. \ref{cor: FCBF-HOCBF} hold (the relative degree of FCBF and HOCBFs with respect to system \eqref{eq:low pass filter} is $m_{a}=1$). Then, the filtered input \( \boldsymbol{u}_f \) is Lipschitz continuous and satisfies the safety requirement as well as the control bounds \( \boldsymbol{u}_{\min} \le \boldsymbol{u}_f \le \boldsymbol{u}_{\max} \) for all \( t \ge 0 \).
\end{corollary}
\begin{proof}
The result follows directly from Cor. \ref{cor: FCBF-HOCBF}.
\end{proof}
 \begin{remark}[Smoothness of Filtered Control]
\label{rem:low pass filter}
Based on Cor. \ref{cor: FCBF-HOCBF}, any low-pass filter can be designed as an auxiliary dynamic system \eqref{eq:smoothness filter} to enhance the smoothness of the filtered input \( \boldsymbol{u}_f \). Low-pass filters are designed to attenuate the high-frequency components of a signal, allowing only the low-frequency (slowly-varying) content to pass through. When applied to control inputs, this filtering process suppresses abrupt changes and sharp transitions, resulting in smoother, more gradual control actions. The degree of smoothing depends on the filter design parameters, such as its order (which determines the steepness of the response) and cutoff frequency (which sets the maximum rate of output variation), although all low-pass filters inherently limit rapid changes in the output. In Cor. \ref{cor:first-order-low pass filter}, a larger $\tau$ in first-order low-pass filter slows down the response of $\boldsymbol{u}_f$, leading to smoother and more gradual changes over time. Conversely, a smaller $\tau$ allows $\boldsymbol{u}_f$ to track $\boldsymbol{\nu}$ more closely, reducing the smoothing effect. If we rewrite \eqref{eq:Lipschitz-proof-3} as
\begin{equation}
\label{eq:Lipschitz-proof-4}
\begin{split}
\left | \dot{\boldsymbol{u}}_{k,f} \right | \le \max (\left | \kappa_{k,1}^{\min}(\boldsymbol{u}_{k,f}-\boldsymbol{u}_{k,\min}) \right |, \\
\left | \kappa_{k,1}^{\max}(\boldsymbol{u}_{k,\max}- \boldsymbol{u}_{k,f}) \right |)\le L_{k}, 
\end{split}
\end{equation}  
where $\alpha_{k,1}^{\min}(\cdot),\alpha_{k,1}^{\max}(\cdot)$ are defined as linear functions with hyperparameters $\kappa_{k,1}^{\min}\ge0,\kappa_{k,1}^{\max}\ge0$. We observe that these hyperparameters affect the Lipschitz constant $L_{k}$, and consequently influence the smoothness of the filtered input \( \boldsymbol{u}_f \). In particular, increasing $\kappa_{k,1}^{\min},\kappa_{k,1}^{\max}$ leads to a larger Lipschitz constant $L_{k}$, which allows for greater instantaneous variation in \( \boldsymbol{u}_f \).
\end{remark}

\subsection{Optimal Control with FCBFs}
Consider the optimal control problem from \eqref{eq:cost-function-1}. Since we need to use filtered input $\boldsymbol{u}_f$ to replace $\boldsymbol{u}$ and introduce auxiliary input $\boldsymbol{\nu}$ to ensure safety and Lipschitz continuity, we reformulate the cost in \eqref{eq:cost-function-1} as
\begin{small}
\begin{equation}
\label{eq:cost-function-2}
\begin{split}
 \min_{\boldsymbol{\nu}} \int_{0}^{T} 
 [D(\left \| \boldsymbol{\nu} \right \| )+p\left \| \boldsymbol{x}(T)-\boldsymbol{x}_{e} \right \| ^{2}]dt,
\end{split}
\end{equation}
\end{small}
where $\left \| \cdot \right \|$ denotes the 2-norm of a vector, $D(\cdot)$ is a strictly increasing function of its argument and $T>0$ denotes the ending time. $p\left \| \boldsymbol{x}(T)-\boldsymbol{x}_{e} \right \| ^{2}$ denotes state convergence similar to \eqref{eq:cost-function-1}. To minimize $\left \| \boldsymbol{u}_{f} \right \|$, we can perform input–output linearization for \eqref{eq:smoothness filter} and use the CLF introduced in ~\cite[Eq. 24]{xiao2021adaptive} to minimize $\boldsymbol{u}_{f}^{\top}\boldsymbol{u}_{f}$. We can then formulate the CLFs (in Def. \ref{def:control-l-f} and  ~\cite[Eq. (24)]{xiao2021adaptive}), HOCBFs (in Thm. \ref{thm:lipschitz-guarantee} and Cor. \ref{cor: FCBF-HOCBF}) and FCBFs (in Def. \ref{def:FCBF} and Cor. \ref{cor: FCBF-HOCBF}) as constraints of the QP with cost function \eqref{eq:cost-function-2} to realize safety-critical control.
\begin{remark}[Parameter-Tuning for FCBFs and HOCBFs]
\label{rem: parameter-tuning}
There are many hyperparameters to tune in FCBFs and HOCBFs to satisfy conditions in Cor. \ref{cor: FCBF-HOCBF}, such as the hyperparameters in class $\kappa$ functions \( \alpha_i \) in FCBFS and HOCBFs as well as the hyperparameters in the input regularization filter \eqref{eq:smoothness filter} (e.g., $\tau$ in the first-order low-pass filter \eqref{eq:low pass filter}). These hyperparameters play a critical role in system performance and are typically determined empirically. However, tuning these parameters is non-trivial and often depends on the specific application. To address this challenge, since the framework of FCBFs is analogous to that of HOCBFs, both can be reformulated as Auxiliary Variable Adaptive CBFs (AVCBFs) as introduced in \cite{liu2025auxiliary}. This allows the design of safety-feasibility criteria and a parameterization method to automatically tune the corresponding hyperparameters.
\end{remark}

\section{Case Study and Simulations}
\label{sec:Case Study and Simulations}

In this section, we consider the unicycle model with the dynamics given by \eqref{eq:UM-dynamics2} for Problem \ref{prob:smooth-prob}, which is more realistic than the simplified unicycle model used in the case study of \cite{xiao2021high2}. We use MATLAB's \texttt{quadprog} solver to solve the QP at each time step, and integrate the system dynamics using the \texttt{ode45} function. All computations are performed on a computer equipped with an Intel\textsuperscript{\textregistered} Core\texttrademark{} i7-11750F CPU @ 2.50\,GHz. The average computation time for each QP is less than 0.01\,s.

\begin{small}
\begin{equation}
\label{eq:UM-dynamics2}
\underbrace{\begin{bmatrix}
\dot{x}(t) \\
\dot{y}(t) \\
\dot{\theta}(t)\\
\dot{v}(t)
\end{bmatrix}}_{\dot{\boldsymbol{x}}(t)}  
=\underbrace{\begin{bmatrix}
 v(t)\cos{(\theta(t))}  \\
 v(t)\sin{(\theta(t))} \\
 0 \\
 0
\end{bmatrix}}_{f(\boldsymbol{x}(t))} 
+ \underbrace{\begin{bmatrix}
  0 & 0\\
  0 & 0\\
  1 & 0\\
  0 & \frac{1}{M} 
\end{bmatrix}}_{g(\boldsymbol{x}(t))}\underbrace{\begin{bmatrix}
   u_{1}(t)   \\
  u_{2}(t) 
\end{bmatrix}}_{\boldsymbol{u}(t)}.
\end{equation}
\end{small}

In \eqref{eq:UM-dynamics2}, $M$ is the mass of the unicycle, $[x, y]^{\top}$ denote the coordinates of the unicycle, $v$ is its linear speed, $\theta$ denotes the heading angle, and $\boldsymbol{u}$ represent the angular velocity ($u_{1}$) and the driven force ($u_{2}$), respectively. We use a first-order low-pass filter \eqref{cor:first-order-low pass filter} as the auxiliary dynamic system and replace the original control input 
$\boldsymbol{u}$ with the filtered input $\boldsymbol{u}_{f}$. This filtered input is then integrated with the unicycle model in \eqref{eq:UM-dynamics2}. The resulting augmented system is given as follows:

\begin{small}
\begin{equation}
\label{eq:UM-dynamics3}
\underbrace{\begin{bmatrix}
\dot{x}(t) \\
\dot{y}(t) \\
\dot{\theta}(t)\\
\dot{v}(t) \\
\dot{u}_{f_1}(t) \\
\dot{u}_{f_2}(t)
\end{bmatrix}}_{(\dot{\boldsymbol{x}}(t),\dot{\boldsymbol{u}}_{f}(t))}
=\begin{bmatrix}
 v(t)\cos{(\theta(t))}  \\
 v(t)\sin{(\theta(t))} \\
 u_{f_1}(t)\\
 \frac{u_{f_2}(t)}{M}\\
 -\frac{1}{\tau} u_{f_1}(t)  \\
 -\frac{1}{\tau} u_{f_2}(t) 
\end{bmatrix} 
+ \begin{bmatrix}
  0 & 0\\
  0 & 0\\
  0 & 0\\
  0 & 0 \\
  \frac{1}{\tau} & 0\\
  0 & \frac{1}{\tau}
\end{bmatrix}\underbrace{\begin{bmatrix}
   \nu_{1}(t)   \\
  \nu_{2}(t) 
\end{bmatrix}}_{\boldsymbol{\nu}(t)}.
\end{equation}
\end{small}
The system is discretized with a time step of \( \Delta t = 0.1\,\text{s} \), and the total number of discretization steps in Problem \ref{prob:smooth-prob} is \( T = 5s \). The augmented unicycle model in \eqref{eq:UM-dynamics3} is subject to the following input constraints:
\begin{equation}
\begin{split}
\label{eq:filtered-input-constraint}
\mathcal{U}_{f}&=\{\boldsymbol{u}_{f}\in \mathbb{R}^{2}: [-5, -5M]^{\top} \le \boldsymbol{u}_{f}\le [5, 5M]^{\top}\}.
\end{split}
\end{equation}
The initial state is \( [-3,\ 0,\ \frac{\pi}{12},\ 2]^{\top} \), which is indicated by a black diamond in Fig.~\ref{fig:FCBF-trajectory}. The circular obstacle is centered at \( (x_{o}, y_{o}) = (0,\ 0) \) with radius \( r_o = 1 \), and is shown in pale red. We consider the case when the unicycle has to avoid this circular obstacle. The distance is considered safe
if $(x-x_{o})^{2}+(y-y_{o})^{2}-r_{o}^{2}\ge 0$. We define HOCBFs to ensure safety as
\begin{equation}
\label{eq:HOCBFs-1}
\begin{split}
&\psi_{0}(\boldsymbol{x})= (x-x_{o})^{2}+(y-y_{o})^{2}-r_{o}^{2},\\
&\psi_{1}(\boldsymbol{x})\coloneqq \dot{\psi}_{0}(\boldsymbol{x})+k_{1}\psi_{0}(\boldsymbol{x}),\\
&\psi_{2}(\boldsymbol{x},\boldsymbol{u}_{f})\coloneqq \dot{\psi}_{1}(\boldsymbol{x},\boldsymbol{u}_{f})+k_{2}\psi_{1}(\boldsymbol{x}),
\end{split}
\end{equation}
where $\alpha_{1}(\cdot),\alpha_{2}(\cdot)$ are defined as linear functions. We define FCBFs as
\begin{equation}
\label{eq:FCBF}
\begin{split}
&\psi_{0,f}(\boldsymbol{x},\boldsymbol{u}_{f})\coloneqq \psi_{2}(\boldsymbol{x},\boldsymbol{u}_{f}),\\
&\psi_{1,f}(\boldsymbol{x},\boldsymbol{u}_{f}, \bm \nu)\coloneqq \dot{\psi}_{0,f}(\boldsymbol{x},\boldsymbol{u}_{f}, \bm \nu)+k_{3}\psi_{0,f}(\boldsymbol{x},\boldsymbol{u}_{f}),
\end{split}
\end{equation}
where $\alpha_{3}(\cdot)$ is defined as a linear function. We define HOCBFs to ensure the filtered input is Lipschitz continuous and satisfies the constraints \eqref{eq:filtered-input-constraint} as
\begin{equation}
\label{eq:HOCBFs-2}
\begin{split}
&\psi_{i,0}^{\text{min}}=u_{f_i}-u_{i,\text{min}},\psi_{i,0}^{\text{max}}=u_{i,\text{max}}-u_{f_i}\\
&\psi_{i,1}^{\text{min}}(u_{f_i},\nu_{i})\coloneqq \dot{\psi}_{i,0}^{\text{min}}(u_{f_i},\nu_{i})+k_{i,1}^{\text{min}}\psi_{i,0}^{\text{min}},\\
&\psi_{i,1}^{\text{max}}(u_{f_i},\nu_{i})\coloneqq \dot{\psi}_{i,0}^{\text{max}}(u_{f_i},\nu_{i})+k_{i,1}^{\text{max}}\psi_{i,0}^{\text{max}},
\end{split}
\end{equation}
where $\alpha_{i,1}^{\text{min}}(\cdot),\alpha_{i,1}^{\text{max}}(\cdot)$ are defined as linear functions and $i\in \{1,2\}$. To satisfy the state convergence, we define a CLF $V(\boldsymbol{x}(t),\boldsymbol{\pi}(t))=(10(\theta(t)-\theta_{d})+u_{f_1}(t)+u_{f_2}(t))^{2}$ with $\theta_{d}=atan2(\frac{y_{d}-y(t)}{x_{d}-x(t)}), c_{1}=c_{2}=1$ to stabilize $\theta(t)$ to $\theta_{d}$ and $u_{f_1}(t)+u_{f_2}(t)$ to $0$ and formulate the relaxed constraint in \eqref{eq:clf} as
\begin{equation}
\label{eq:ACC-clf-2}
\dot{V}(\boldsymbol{x}(t),\boldsymbol{u}_{f}(t), \bm \nu(t)) +c_{3}V(\boldsymbol{x}(t),\boldsymbol{u}_{f}(t))\le \delta(t), 
\end{equation}
where $\delta(t)$ is a relaxation that makes \eqref{eq:ACC-clf-2} a soft constraint, $(x_{d}, y_{d})=(1.5, 0)$ is the desired location and $r_{d}=0.1$ is tolerance for state convergence (indicated by a green circle in Fig.~\ref{fig:FCBF-trajectory}).
By formulating constraints from HOCBFs \eqref{eq:HOCBFs-1},\eqref{eq:HOCBFs-2}, FCBFs \eqref{eq:FCBF}, CLFs \eqref{eq:ACC-clf-2}, we can define the cost function 
 for the QP as
 
 \begin{small}
\begin{equation}
\label{eq:AVBCBF-cost-2}
\begin{split}
\min_{\boldsymbol{\nu}(t),\delta(t)} \int_{0}^{T}[\boldsymbol{\nu}(t)^{\top}\boldsymbol{\nu}(t)+Q\delta(t)^{2}]dt.
\end{split}
\end{equation}
\end{small}
Other parameters are $c_{3}=10, M=1650kg, u_{f_1}(0)=u_{f_2}(0)=0, Q=10^{5}, k_{1}=k_{2}=10, k_{1,1}^{\text{min}}=k_{2,1}^{\text{min}}=k_{1,1}^{\text{max}}=k_{2,1}^{\text{max}}=\alpha$. The above method is referred to as FCBF. For the first benchmark, we only include constraints from HOCBFs \eqref{eq:HOCBFs-1}, input \eqref{eq:filtered-input-constraint} and CLFs \eqref{eq:ACC-clf-2} with $V(\boldsymbol{x}(t))=(\theta(t)-\theta_{d})^{2}$ and refer to this as HOCBF method (note that the decision variables in the cost function \eqref{eq:AVBCBF-cost-2} are $\boldsymbol{u}(t),\delta(t)$). For the second benchmark, we build upon the first benchmark by adding $0.1(u_{1}(t)-u_{1}(t-\Delta t))^{2}+0.1(u_{2}(t)-u_{2}(t-\Delta t))^{2}$ to the cost function. This term encourages gradual changes by penalizing the difference between the control inputs at consecutive time steps. We refer to the second benchmark as sp-HOCBF (smoothness-penalized HOCBF).

\begin{figure}[ht]
\vspace*{3mm}
    \centering
    \includegraphics[scale=0.15]{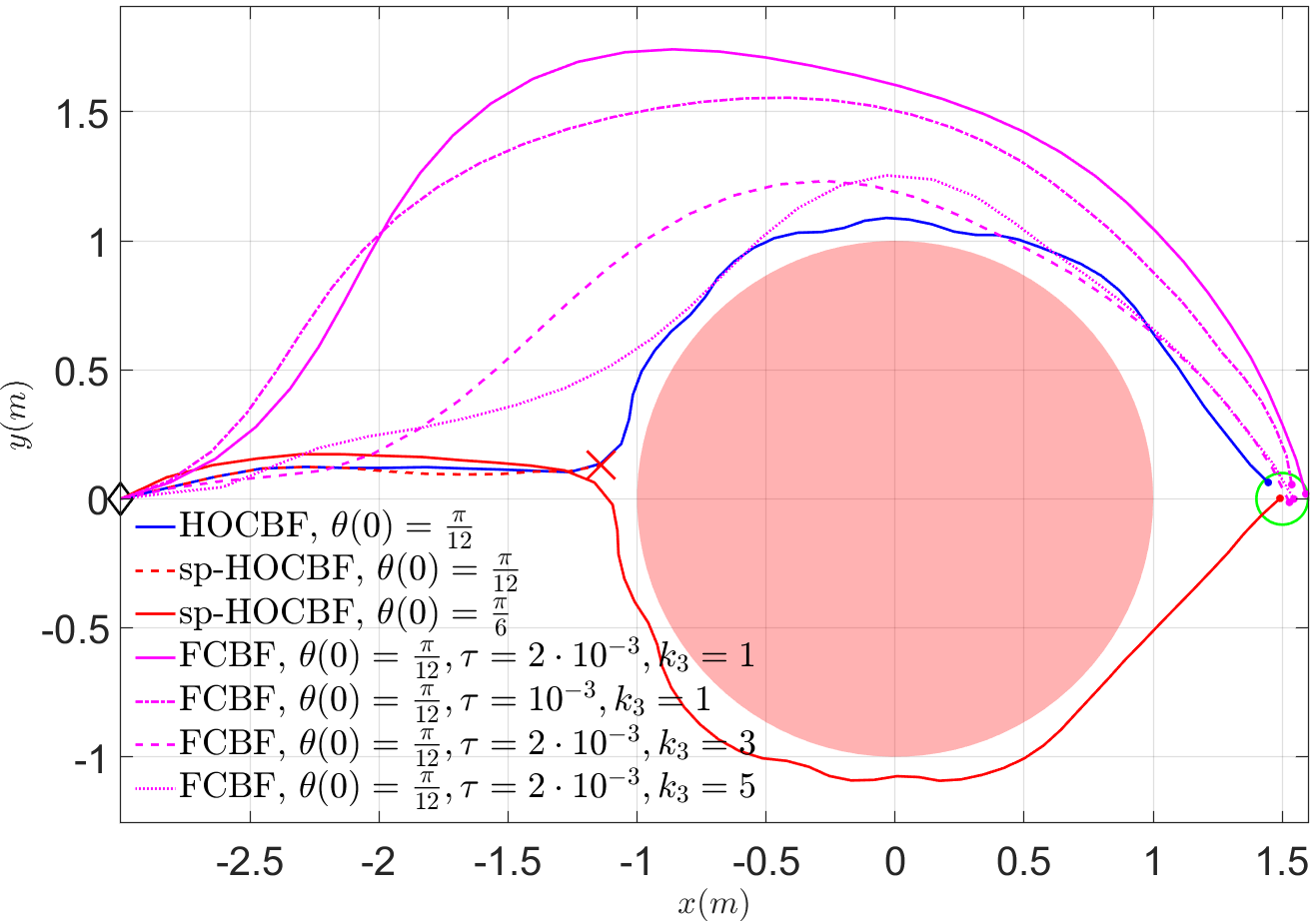}
    \caption{Closed-loop trajectories with controllers derived using FCBF (magenta), HOCBF (blue) and sp-HOCBF (red). FCBF ($\alpha=1$) and HOCBF perform well in safety-critical navigation when the initial heading angle is small. FCBF produces smoother trajectories, demonstrates strong adaptability, and maintains feasibility across a range of hyperparameters.}
    \label{fig:FCBF-trajectory}
\end{figure} 

As shown in Fig. \ref{fig:FCBF-trajectory}, FCBF, HOCBF, and sp-HOCBF all successfully guide the unicycle from the initial position (indicated by the black diamond) to the desired area (green circle) while ensuring safety. Due to the constraint on the difference between consecutive control inputs in sp-HOCBF, the QP becomes infeasible when the initial heading angle is small (e.g., $\theta(0)=\frac{\pi}{12}$), as indicated by the red cross. Increasing the initial heading angle to $\frac{\pi}{6}$ resolves the infeasibility issue in sp-HOCBF, while the other controllers (FCBF and HOCBF) remain feasible even under small heading angles, demonstrating better feasibility. When the FCBF hyperparameter $k_{3}$ is increased (the other hyperparameters are held constant, e.g., $\alpha=1$, $\tau=2\cdot 10^{-3}$), the unicycle follows a more aggressive control strategy and moves closer to the obstacle. Among the three methods, FCBF produces the smoothest trajectory.

As illustrated in Figs. \ref{fig:FCBF-input1} and \ref{fig:FCBF-input2}, and consistent with Rem. \ref{rem:low pass filter}, increasing the hyperparameter $\tau$ tends to result in smoother variations in the control input $\boldsymbol{u}$ for FCBF when other hyperparameters are held constant (e.g., $k_{3}=\alpha =1$). Compared to HOCBF and sp-HOCBF, the control input under FCBF exhibits smoother variations. However, by comparing the input curves of HOCBF and sp-HOCBF, we observe that introducing a smoothness penalty does not significantly improve the smoothness of the control input. As shown in Fig. \ref{fig:FCBF-input3}, when the other hyperparameters are held constant, e.g., $k_{3}=1$, $\tau=2\cdot 10^{-3}$, smaller values of 
$\alpha$ tend to promote smoother variations in $\boldsymbol{u}$ for FCBF, which is consistent with Rem. \ref{rem:low pass filter}. We conclude that FCBF offers improved input smoothness and greater control flexibility compared to HOCBF and sp-HOCBF.

\begin{figure}[ht]
\vspace*{3mm}
    \centering
    \includegraphics[scale=0.44]{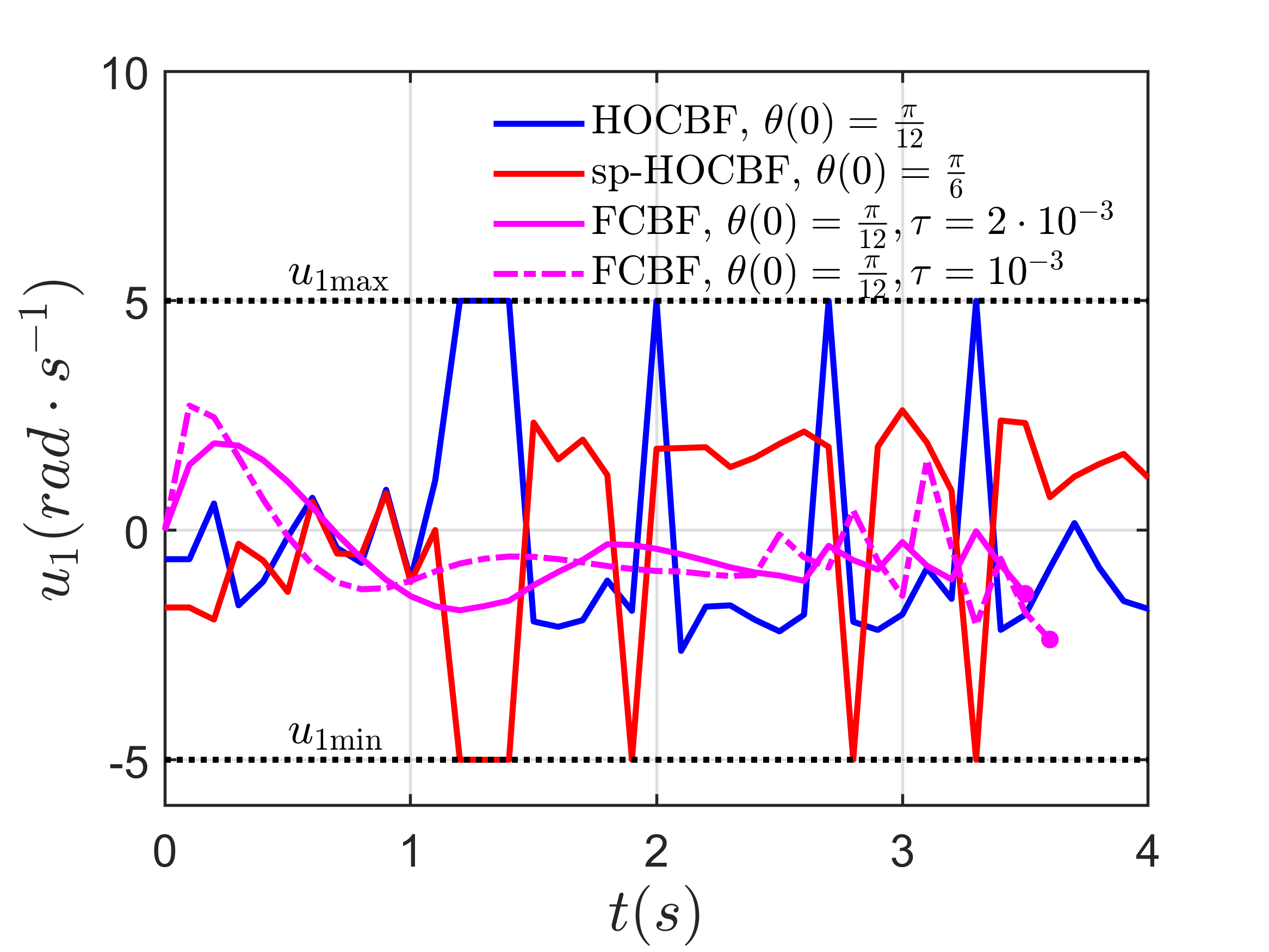}
    \caption{Control input $u_{1}$ (angular velocity) over time with different controllers. FCBF ($k_{3}=\alpha=1$) ensures smoother transitions of $u_{1}$ compared to HOCBF and sp-HOCBF.}
    \label{fig:FCBF-input1}
\end{figure} 
\begin{figure}[ht]
    \centering
    \includegraphics[scale=0.44]{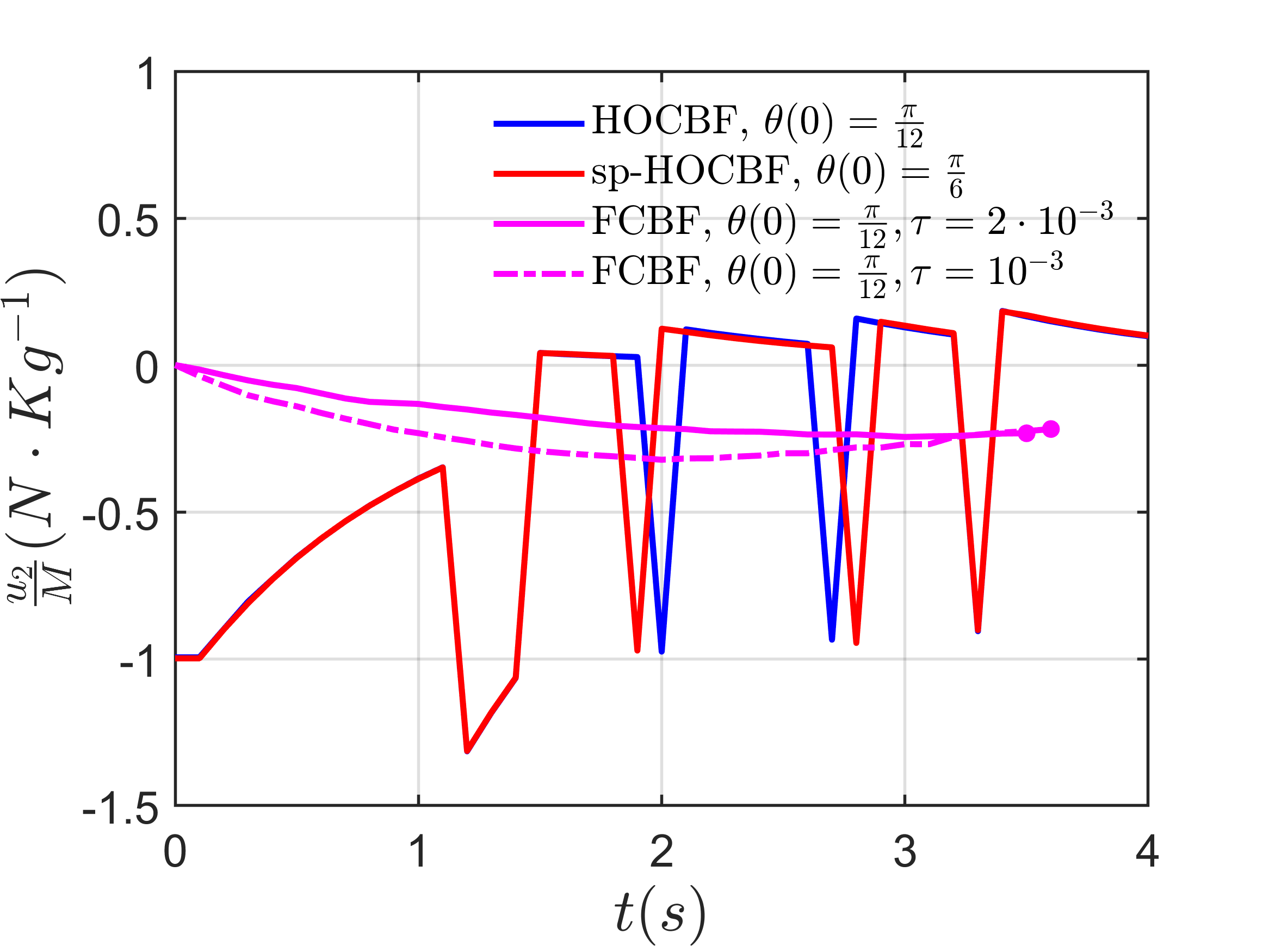}
    \caption{Control input $u_{2}$ (driven force) over time with different controllers. FCBF ($k_{3}=\alpha=1$) ensures smoother transitions of $u_{2}$ compared to HOCBF and sp-HOCBF.}
    \label{fig:FCBF-input2}
\end{figure} 

\begin{figure}[t]
    \vspace{3mm}
    \centering
    \begin{subfigure}[t]{0.49\linewidth}
        \centering
        \includegraphics[width=1\linewidth]{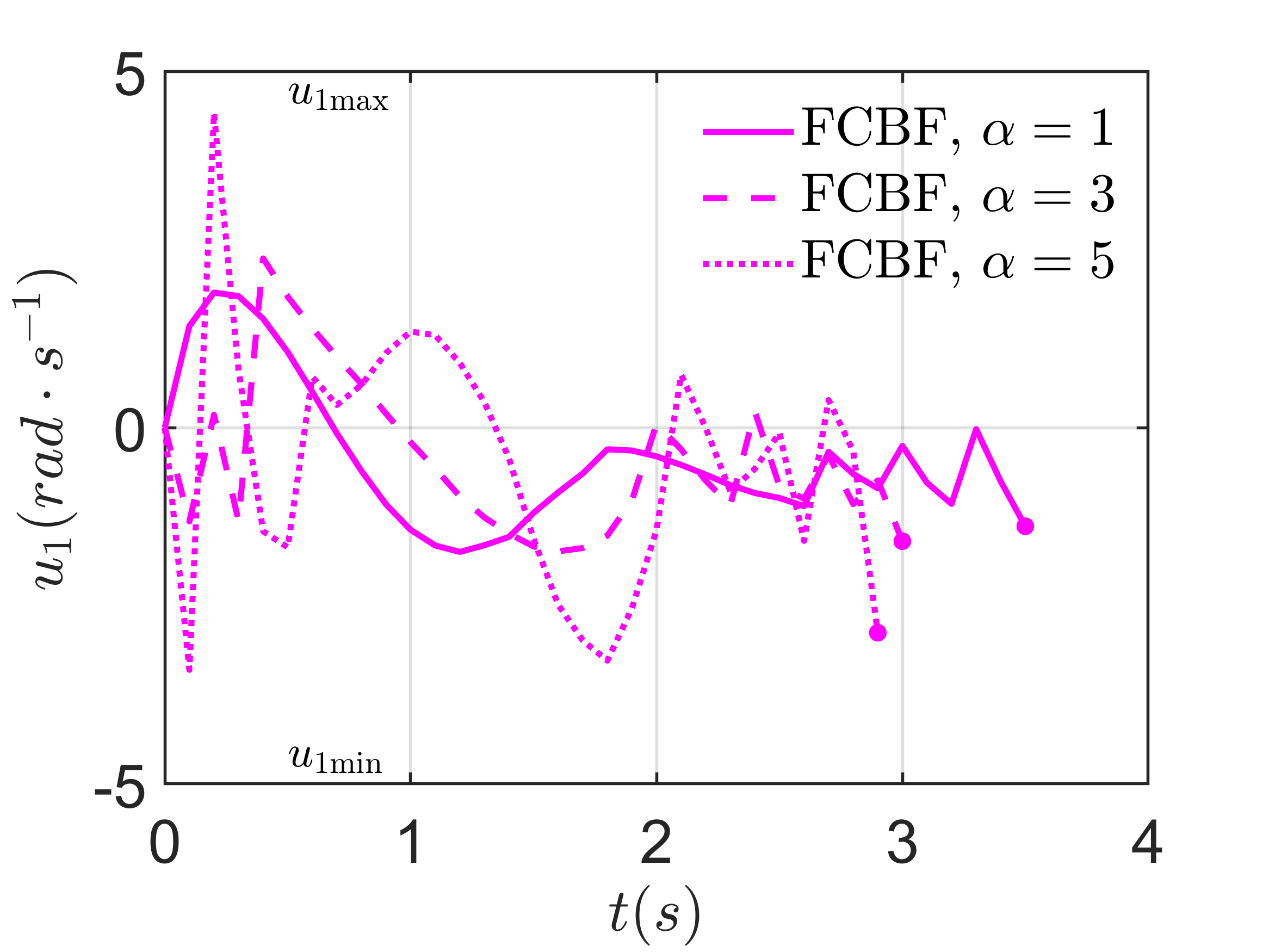}
        \caption{Control input $u_{1}$ over time}
        \label{fig:diff q1}
    \end{subfigure}
    \begin{subfigure}[t]{0.49\linewidth}
        \centering
        \includegraphics[width=1\linewidth]{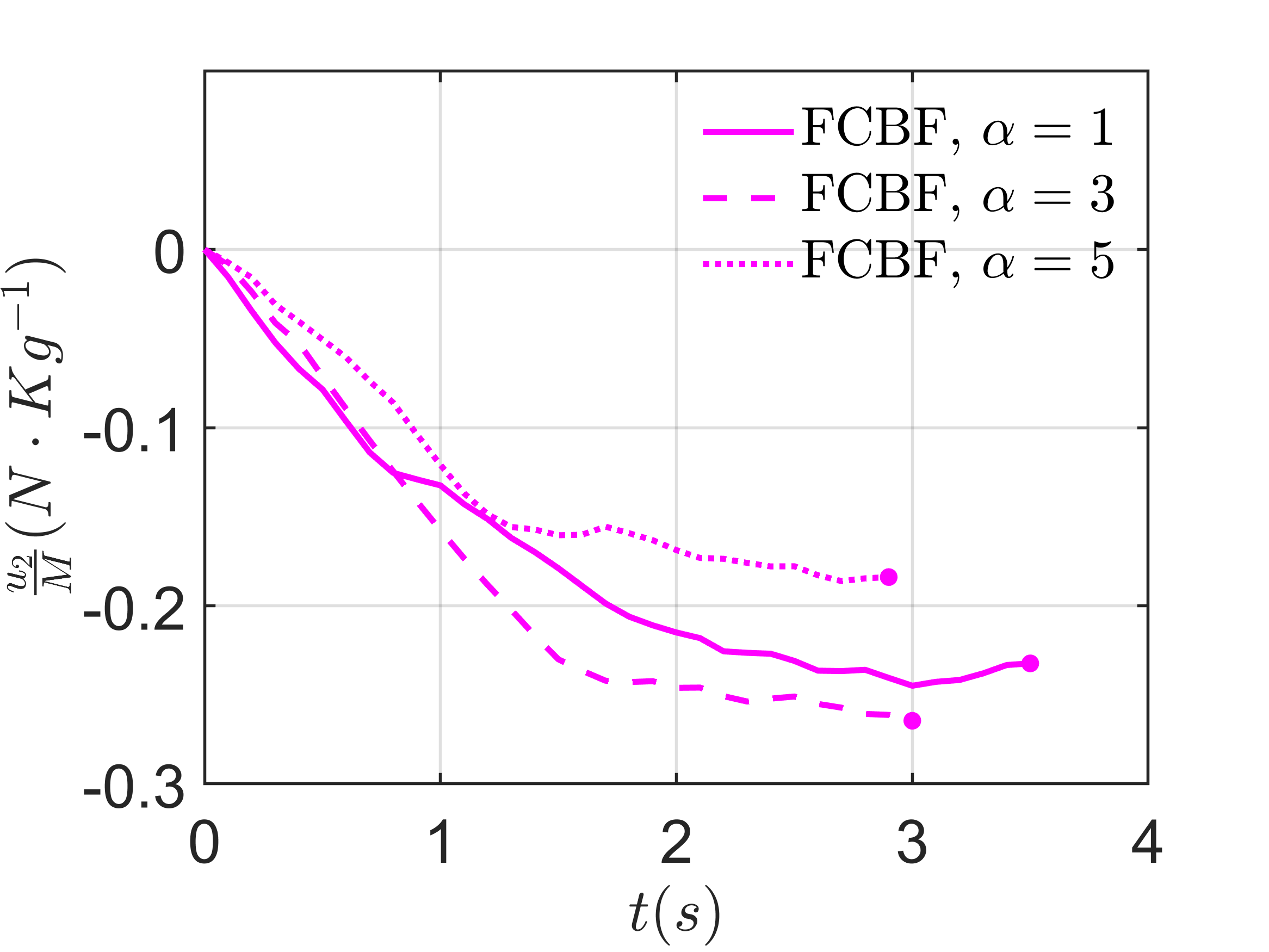}
        \caption{Control input $u_{2}$ over time}
        \label{fig:diff q12}
    \end{subfigure}
    \caption{FCBF ($k_{3}=1$, $\tau=2\cdot 10^{-3}$) with different class $\kappa$ function hyperparameter ($\alpha$) is evaluated. Smaller $\alpha$ tends to promote smoother variations in $\boldsymbol{u}$.} 
    \label{fig:FCBF-input3}
\end{figure}

\section{Conclusion and Future Work}
\label{sec:conclusion}
This paper proposes a new framework, Filtered Control Barrier Functions (FCBFs), that addresses the challenge of generating Lipschitz continuous control inputs in safety-critical systems. By introducing an auxiliary dynamic system acting as a low-pass filter, we guarantee that the resulting filtered input remains Lipschitz continuous and satisfies control input limitations while preserving the safety guarantees of traditional HOCBFs. The formulation retains the structure and computational efficiency of a single QP, making it amenable to real-time implementation. Simulation results demonstrate that FCBFs provide smoother and more flexible control trajectories compared to existing methods such as HOCBF and smoothness-penalized HOCBF. Future work will explore learning-based hyperparameter tuning methods for stochastic systems, with implementation on physical robotic platforms.
\label{sec:Conclusion and Future Work}

\bibliographystyle{IEEEtran}
\balance
\bibliography{references.bib}
\end{document}